\documentclass[a4paper]{article}
\usepackage{authblk}
\usepackage[T1]{fontenc}
\usepackage{hyperref}
\usepackage{graphicx}
\usepackage{color}

\usepackage{enumerate}
\usepackage{yfonts,array}
\usepackage{amsmath,amssymb,units,stmaryrd}
\usepackage{cleveref}
\usepackage{upgreek,xcolor,bm,tikz,stackrel}
\usepackage[inline, shortlabels]{enumitem}
\usepackage{graphicx}
\usepackage[all]{xy}
\usepackage{tikz-cd}
\usetikzlibrary{decorations.pathmorphing}
\usepackage{wasysym}
\usetikzlibrary{arrows}
\usepackage{algorithm}% http://ctan.org/pkg/algorithms

\usepackage{amsthm}
\newtheorem{definition}{Definition}
\newtheorem{corollary}{Corollary}

\usepackage{algpseudocode}

\newcommand{\define}[1]{\textbf{#1}}

\newcommand{\lm}{\vec{\mathcal  Q}}
\newcommand{\domlm}{|\vec {\mathcal  Q}|}

\newcommand{\imp}{\mathop \Rightarrow}
\newcommand{\dimp}{\mathop \Leftarrow}

\newcommand{\gtl}{{\sf GTL}}

\newcommand{\lanfull}{{\mathcal L}}

\newcommand{\ignore}[1]{}

\newcommand{\lgt}[1]{\|#1\|}
\newcommand{\nx}{{\ocircle}}
\newcommand{\nec}{\Box}
\newcommand{\ps}{\Diamond}
\newcommand{\val}[1]{\lb #1 \rb}

\newcommand{\moments}[1]{\mathbb M_{ #1 }}

\newcommand{\type}[1]{\mathbb T_{ #1 }}

\def\lb{\left\llbracket}
\def\rb{\right\rrbracket}
\def\<{\left (}

\def\>{\right )}
\def\({\left (}
\def\){\right )}

\def\cbra{\left \{}
\def\cket{\right \}}

\DeclareSymbolFont{AMSb}{U}{msb}{m}{n}
\DeclareMathSymbol{\N}{\mathbin}{AMSb}{"4E}
\DeclareMathSymbol{\Z}{\mathbin}{AMSb}{"5A}
\DeclareMathSymbol{\R}{\mathbin}{AMSb}{"52}
\DeclareMathSymbol{\Q}{\mathbin}{AMSb}{"51}
\DeclareMathSymbol{\I}{\mathbin}{AMSb}{"49}
\newtheorem{thm}{Theorem}[section]
\newtheorem{lem}[thm]{Lemma}
\newtheorem{prop}[thm]{Proposition}

\begin{document}

\title{Time and G\"odel: Fuzzy temporal reasoning in PSPACE}
\author[1,2]{Juan Pablo Aguilera \footnote{\href{mailto:juan.aguilera@UGent.be}{\tt juan.aguilera@UGent.be}}}
\author[3]{Mart\'in Di\'eguez \footnote{\href{mailto:martin.dieguezlodeiro@univ-angers.fr}{\tt martin.dieguezlodeiro@univ-angers.fr}}}
\author[1,4]{David Fern\'andez-Duque \footnote{\href{mailto:david.fernandezduque@ugent.be}{\tt david.fernandezduque@ugent.be}}}
\author[1]{Brett McLean \footnote{\href{mailto:brett.mclean@ugent.be}{\tt brett.mclean@ugent.be}}}

\affil[1]{\centering Department of Mathematics WE16, Ghent University, Ghent, Belgium}
\affil[2]{\centering Institute of Discrete Mathematics and Geometry, Vienna University of Technology, Vienna, Austria}
\affil[3]{\centering LERIA, University of Angers, Angers, France}
\affil[4]{\centering ICS of the Czech Academy of Sciences, Prague, Czech Republic}

%% required for running head on odd and even pages, use suitable
%% abbreviations in case of long titles and many authors:

%\subjclass{F.4.1 Mathematical Logic}
%\titlecomment{OPTIONAL comment concerning the title, \eg, if a variant
%or an extended abstract of the paper has appeared elsewehere}
%%%%%%%%%%%%%%%%%%%%%%%%%%%%%%%%%%%%%%%%%%%%%%%%%%%%%%%%%%%%%%%%%%%%%%%%%%%

%% the abstract has to PRECEED the command \maketitle (LMCS) or FOLLOW the \maketitle (LICS)
%% be sure not to issue the \maketitle command twice!

\maketitle

\begin{abstract}We investigate a non-classical version of linear temporal logic whose propositional fragment is G\"odel--Dummett logic (which is well known both as a superintuitionistic logic and a t-norm fuzzy logic). We define the logic using two natural semantics, a real-valued semantics and a bi-relational semantics, and show that these indeed define one and the same logic. Although this G\"odel temporal logic does not have any form of the finite model property for these two semantics, we show that every falsifiable formula is falsifiable on a finite quasimodel, which yields decidability of the logic. We then strengthen this result by showing that this G\"odel temporal logic is {\sc pspace}-complete.
\end{abstract}

%% mandatory lists of keywords and classifications:
%\keywords{G\"odel--Dummett logic; linear temporal logic; intuitionistic logic; fuzzy logic; PSPACE-complete}

%% start the paper here:
\section{Introduction}

The importance of temporal logics and, independently, of fuzzy logics in computer science is well established.
The potential usefulness of their combination is clear: for instance, it would provide a natural framework for the specification of programs dealing with vague data.
Sub-classical temporal logics have mostly been studied in the context of here-and-there logic, which allows for three truth values and is the basis for temporal answer set programming~\cite{taspa,taspb,BalbianiDieguezJelia}.

One may, however, be concerned that infinite-valued temporal logics could lead to an explosion in computational complexity, as has been known to happen when combining fuzzy logic with transitive modal logics: these combinations are often undecidable~\cite{Vidal21}, or decidable with only an exponential upper bound being known \cite{BalbianiDF21}.
As we will see, this need not be the case: the combination of G\"odel--Dummett logic with linear temporal logic, which we call G\"odel temporal logic ($\sf GTL$), remains {\sc {\sc pspace}}-complete, the minimal possible complexity given that classical $\sf LTL$ embeds into it.
This is true even when the logic is enriched with the dual implication~\cite{CecyliaRauszer1980}, which has been argued in \cite{BilkovaFK21} to be useful for reasoning with incomplete or inconsistent information.

The decidability of $\sf GTL$ is already surprising, as it does not enjoy the finite model property.
In fact, $\sf GTL$ possesses two natural semantics, corresponding to whether it is viewed as a fuzzy logic or a superintuitionistic logic.
As a fuzzy logic, propositions take values in $[0,1]$, and truth values of compound propositions are defined using standard operations on the real line.
As a superintuitionistic logic, models consist of bi-relational structures equipped with a partial order to interpret implication intuitionistically and a function to interpret the $\sf LTL$ tenses.
Remarkably, the two semantics give rise to the same set of valid formulas, which should provide two different avenues to prove decidability of $\sf GTL$ via the finite model property.
Unfortunately, as we will see, $\sf GTL$ does not enjoy the finite model property for either of these semantics.

Thus we instead introduce quasimodels, which do enjoy their own version of the finite model property.
Quasimodels are not `true' models in that the functionality of the `next' relation is lost, but they give rise to standard bi-relational models by unwinding.
Similar structures were used to prove upper complexity bounds for dynamic topological logic \cite{Fernandez09,dtlaxiom} and intuitionistic temporal logic \cite{F-D18}, but they are particularly effective in the setting of G\"odel temporal logic, as they yield an optimal {\sc {\sc pspace}} upper bound.

\subsubsection*{Structure of paper} In \Cref{SecBasic} we introduce the temporal language that we work with, and then introduce both the real semantics and bi-relational semantics for G\"odel temporal logic. In \Cref{sec:real} we prove  the equivalence of these two semantics, that is, that they yield the same validities. In \Cref{SecGen} we first note that we do not have a finite model property for either of these semantics. But then we define quasimodels, and in later sections we show that our G\"odel temporal logic is sound and complete for the class of finite quasimodels. In \Cref{SecGen} we show that G\"odel temporal logic is sound for \emph{all} quasimodels, constructing a bi-relational model from an arbitrary quasimodel by unwinding selected paths within the quasimodel. In \Cref{Sec:quotient}, given a bi-relational model falsifying a formula, we describe how to produce a finite (exponential in the length of the formula) quasimodel also falsifying the formula. This proves completeness of G\"odel temporal logic for finite quasimodels and the decidability of G\"odel temporal logic.
Finally, in \Cref{Sec:PSPACE} we refine this decidability result, showing that G\"odel temporal logic is in fact {\sc pspace}-complete.%\david{I corrected repeated instances of ``Class $\mathcal C$ is sound/complete for logic $\Lambda$'', whereas standard usage is ``Logic $\Lambda$ is sound/complete for class $\mathcal C$''.}

\section{Syntax and Semantics}\label{SecBasic}

In this section we first introduce the temporal language we work with and then two possible semantics for this language: real semantics and bi-relational semantics. 

Fix a countably infinite set $\mathbb P$ of propositional variables. Then the \define{G\"odel temporal language} $\lanfull$ is defined by the grammar (in Backus--Naur form):
\[\varphi,\psi := \ \  \bot \  | \   p  \ |  \ \varphi\wedge\psi \  |  \ \varphi\vee\psi  \ |  \ \varphi\imp \psi  \ |  \  \ \varphi\dimp \psi  \ |  \ \nx\varphi \  | \  \ps\varphi \  |  \ \nec\varphi , \]
where $p\in \mathbb P$. Here, $\nx$ is read as `next', $\ps$ as `eventually', and $\nec$ as `henceforth'. The connective $\dimp$ is coimplication and represents the operator dual to implication \cite{Wolter1998}. %\david{Removed the shorthands, I don't think they're used.}
 We also use $\neg\varphi$ as a shorthand for $\varphi\imp \bot$.%\brett{$\neg$ is used in \Cref{Sec:PSPACE}.}% and $\varphi\iiff \psi$ as a shorthand for $(\varphi\imp \psi) \wedge (\psi\imp\varphi)$.

%Denote the set of subformulas of $\varphi\in\lanfull$ by ${\mathrm{sub}}(\varphi)$, and the length of $\varphi$ (or, more precisely, $\# {\mathrm{sub}}(\varphi)$) by $\lgt\varphi$.{\color{red} these notations are not used at the moment}

We now introduce the first of our semantics for the G\"odel temporal language: real semantics, which views $\lanfull$ as a fuzzy logic (enriched with temporal modalities). In the definition, $[0,1]$ denotes the real unit interval.

\begin{definition}[real semantics]\label{DefRSem}
A \define{flow} is a pair $\mathcal T = (T,S)$,
where $ T $ is a set and $S \colon T \to T$ is a function.
A \define{real valuation} on $\mathcal T$ is a function $V \colon \lanfull \times T \to \mathbb [0,1]$ such that, for all $t\in T$, the following equalities hold.
\[
\begin{array}{rclrcl}
V(\bot,t)&=& 0 \\
V(\varphi\wedge\psi,t) &=& \min \{ V(\varphi ,t) , V( \psi,t) \}\hspace{1cm} V(\varphi \vee \psi,t) = \max \{ V(\varphi ,t) , V( \psi,t) \}\\
V(\varphi \imp \psi,t) &=&\begin{cases}
1&\text{if }V(\varphi, t)\leq V(\psi, t)\\
V(\psi,t )&\text{if }V(\varphi,t ) > V(\psi,t )\\
\end{cases}
\\
V(\varphi \dimp \psi,t) &=&\begin{cases}
V(\varphi,t )&\text{if }V(\varphi,t ) > V(\psi,t )\\
0&\text{if }V(\varphi, t)\leq V(\psi, t)
\end{cases}
\\
V(\nx\varphi, t)&=&V( \varphi,S(t))\\
V(\ps\varphi,t)&=& \sup_{n<\omega}  V(\varphi,S^n(t))\hspace{1cm}
V(\nec\varphi,t)= \inf_{n<\omega}  V(\varphi,S^n(t))\\
\end{array}
\]
A flow $\mathcal  T$ equipped with a valuation $V$ is a \define{real (G\"odel temporal) model}.
\end{definition}

The second semantics, bi-relational semantics, views $\lanfull$ as an intuitionistic logic (temporally enriched).

\begin{definition}[bi-relational semantics]\label{DefKSem}
A \define{(G\"odel temporal) bi-relational frame} is a quadruple $\mathcal  F=(W,T,{\leq},S)$ where $(W,{\leq})$ is a linearly ordered set and $(T,S)$ is a flow.
A \define{bi-relational valuation} on $\mathcal  F$ is a function $\lb\cdot\rb\colon\lanfull \to 2^{W\times T}$ such that, for each $p \in \mathbb P$, the set $\lb p \rb$ is \emph{downward closed} in its first coordinate, and the following equalities hold.
\[
\begin{array}{rclrcl}
\lb\bot\rb&=&\varnothing \hspace{1cm}
\lb\varphi\wedge\psi\rb =\lb\varphi\rb\cap \lb\psi\rb\hspace{1cm}
\lb\varphi\vee\psi\rb =\lb\varphi\rb\cup \lb\psi\rb\\
\lb\varphi\imp\psi\rb &=& \{ (w,t) \in W\times T \mid \forall v\leq w ((v,t)\in \lb\varphi \rb  
%\\&&\hspace{2.75cm}
\text{ implies } (v,t)\in \lb\psi \rb ) \}\\
\lb\varphi\dimp\psi\rb &=& \{ (w,t) \in W\times T \mid \exists v\geq w ((v,t)\in \lb\varphi \rb  
%\\&&\hspace{2.75cm}
\text{ and } (v,t)\notin \lb\psi \rb ) \}\\
\val{\nx\varphi}&=&(\mathrm{id}_W \times S)^{-1} \val\varphi\\
\val{\ps\varphi}&= &\bigcup_{n<\omega}(\mathrm{id}_W \times S)^{-n} \val\varphi\hspace{1cm}
\val{\nec\varphi}= \bigcap_{n<\omega}(\mathrm{id}_W \times S)^{-n}\val\varphi  \\
\end{array}
\]
A bi-relational frame $\mathcal  F$ equipped with a valuation $\lb\cdot\rb$ is a \define{(G\"odel temporal) bi-relational model}.
\end{definition}

This semantics combines standard semantics for the implications based on $\leq$ (read downward) and for the tenses based on $S$: for example, $(w,t)\in \val{\ps\varphi}$ if and only if there exists $n\geq 0$ such that $(w,S^n(t)) \in \val\varphi$.
Note that, by structural induction, the valuation of \emph{any} $\varphi \in \lanfull$ is downward closed in its first coordinate, in the sense that if $(w,t)\in \val\varphi$ and $v\leq w$, then $(v,t)\in \val \varphi$.

The attentive reader may notice that (with respect to either semantics) $\bot$ is expressible as $\varphi \dimp \varphi$, using any $\varphi$. We choose to include $\bot$ as a primitive symbol for the benefit of those who are interested in the $\dimp$-free language, i.e. the language whose propositional fragment is the familiar language of intuitionistic logic.

Validity of $\lanfull$-formulas is defined in the usual way.

\begin{definition}[validity]
Given a real model $\mathcal  X = (T, S, V)$ and a formula $\varphi\in \lanfull$, we say that $\varphi$ is \define{globally true} on $\mathcal  X$, written $\mathcal  X\models\varphi$, if for all $t \in T$ we have $V(\varphi, t) = 1$.
Given a bi-relational model $\mathcal  X = (\mathcal F, \lb\cdot\rb)$ and a formula $\varphi\in \lanfull$, we say that $\varphi$ is \define{globally true} on $\mathcal  X$, written $\mathcal  X\models\varphi$, if $\val\varphi =W \times T$. 

If $\mathcal  X$ is a flow or a bi-relational frame, we write $\mathcal  X\models\varphi$  and say $\varphi$ is \define{valid} on $\mathcal X$, if $\varphi$ is globally true for every valuation on $\mathcal  X$. If $\Omega$ is a class of flows, frames, or models, we say that $\varphi\in\lanfull$ is \define{valid} on $\Omega$ if, for every $\mathcal  X\in \Omega$, we have $\mathcal  X\models\varphi$. If $\varphi$ is not valid on $\Omega$, it is \define{falsifiable} on $\Omega$.
\end{definition} 

We define the logic $\gtl_{\mathbb R}$ to be the set of  $\lanfull$-formulas that are valid over the class of all flows and the logic $\gtl$ to be the set of $\lanfull$-formulas that are valid over the class of all bi-relational frames.

\section{Real versus~bi-relational validity}\label{sec:real}

In this section, we show that an arbitrary $\lanfull$-formula is real valid if and only if it is bi-relationally valid. That is, $\gtl_{\mathbb R} = \gtl$. This equivalence will be immediate from Lemma \ref{LemmaRealToKripke} and Lemma \ref{LemmaKripkeToReal} below.

\begin{lem}\label{LemmaRealToKripke}
Suppose that $\varphi$ is an $\lanfull$-formula that is not real valid. Then $\varphi$ is not bi-relationally valid.
\end{lem}
\begin{proof}
Let $(T, S)$ be a flow, $V$ a real valuation on $(T,S)$, and $t_0 \in T$ be such that $V(\varphi, t_0) < 1$. Since we are only concerned with the valuation at $t_0$, we may assume without loss of generality that $T = \mathbb{N}$, $t_0 = 0$, and $S$ is the successor function; in particular, that $T$ is countable. 
Let $X_0$ be the set of all real numbers $x$ such that $V(\psi, t) = x$ for some $\lanfull$-formula $\psi$ and some $t\in T$. Thus, $X_0$ is a countable subset of $[0,1]$. Let $X = (0,1) \setminus X_0$. 

We consider the bi-relational frame $\mathcal{F} = (X, T, \leq, S)$, where $\leq$ is the usual order on real numbers, and the bi-relational valuation $\lb \cdot \rb$ given by 
\begin{equation} \label{eqInductionRealToKripke}
(x, t) \in \lb  p \rb \leftrightarrow V(p, t)> x.
\end{equation}
We prove by induction that \eqref{eqInductionRealToKripke} holds for arbitrary $\lanfull$-formulas $\psi$ and arbitrary $x \in X$ and $t \in T$. If so, then letting $x\in X$ be such that $V(\varphi, t_0)< x$ (this exists because $X_0$ is countable), we have $(x, t_0) \not\in \lb \varphi\rb$, so that $(\mathcal{F}, \lb \cdot\rb)$ is a bi-relational countermodel for $\varphi$, as needed.

Let us just verify the cases of $\Box\psi$ and $\chi \dimp \psi$ of the induction. For the first one, suppose $x \in X$ and $t\in T$. Then
\begin{align*}
(x,t) \in \lb \Box\psi \rb 
\leftrightarrow \forall n\in\mathbb{N}\, \big((x,S^n(t)) \in \lb \psi \rb \big)
&\leftrightarrow \forall n\in\mathbb{N}\, \big( V(\psi, S^n(t)) > x \big)\\
&\leftrightarrow \inf_{n\in\mathbb{N}} \{V(\psi, S^n(t)) \} > x\\
&\leftrightarrow V(\Box\psi, t) > x.
\end{align*}
Here, the first and last equivalences follow immediately from the definitions, the second equivalence uses the induction hypothesis, and the left-to-right direction of the third equivalence uses the fact that 
$V(\Box\psi, t) = \inf_{n\in\mathbb{N}} \{V(\psi, S^n(t)) \}$
is not an element of $X$ and thus cannot be equal to $x$.

For the case of $\chi\dimp \psi$, again suppose $x \in X$ and $t\in T$. We have 
\begin{align*}
(x,t) \in \lb \chi\dimp \psi\rb 
&\leftrightarrow \exists y\geq x\, \big( (y,t) \in \lb \chi\rb \wedge (y,t) \not\in \lb\psi\rb\big)\\
&\leftrightarrow \exists y\geq x\, \big( V(\chi, t) > y \wedge V(\psi, t) \leq y \big)\\
&\leftrightarrow V(\chi, t) > V(\psi, t) \wedge V(\chi, t) > x\\
&\leftrightarrow V(\chi\dimp \psi, t) > x,
\end{align*}
as desired. Here, the first and last equivalences are immediate from the definitions, the second one follows from the induction hypothesis, and the third one follows from usual properties of real numbers.
\end{proof}

\begin{lem}\label{LemmaKripkeToReal}
Suppose that $\varphi$ is an $\lanfull$-formula that is not bi-relationally valid. Then $\varphi$ is not real valid.
\end{lem}
\begin{proof}
Suppose that there is a bi-relational frame $\mathcal{F} = (W, T, \allowbreak\leq, S)$ and a valuation $\lb \cdot \rb$ such that $(w,t) \not\in \lb\varphi\rb$. Since we are only concerned with the valuation of $\varphi$ in $(w,t)$, we may as well assume that $T = \mathbb{N}$, that $S$ is the successor function, and that $t = 0$.
By a routine L\"owenheim--Skolem-type argument, we may assume that $W$ is countable.\footnote{Build a suborder $(W^*,{\leq}^*)$ of $(W,{\leq})$ by induction on the structure of $\varphi$ as follows: start with $\{w\}$ and inductively decompose $\varphi$ according to its outermost connective. When considering subformulas $\psi$ of the form $\varphi_0 \imp \varphi_1$ or $\varphi_0 \dimp \varphi_1$, we add to $W^*$ (if necessary), for each world $w^*$ in the suborder being built and for each $s \in T$, a new element of $W$  witnessing the quantifier in the definition of $\lb\psi\rb$ in a way that ensures that ``$(w^*, s) \in \lb\psi\rb$'' holds in $(W,T,{\leq}, S)$ if and only if it holds in $(W^*, T, \leq^*, S)$.}

We define a binary relation on $\lanfull \times T$ by 
\[(\psi, t_1) \leq (\chi, t_2) \leftrightarrow \forall w\in W\, (w, t_1) \in \lb \psi \rb \to (w, t_2) \in \lb \chi \rb,\]
i.e., if the valuation of $\psi$ in $t_1$ is contained in that of $\chi$ in $t_2$.
Since valuations of formulas are downward closed in their first coordinates, this reflexive transitive relation is total (any two elements are comparable). Let $(L, \leq)$ denote the linear order of all equivalence classes $[\psi, s]$ under $\leq$, and let $0$ and $1$ denote the equivalence classes of $[\bot, s]$ and $[\psi \imp \psi, s]$, respectively (this is independent of the choice of $s$ and $\psi$).

We claim that this linear order respects valuations, in the sense that it satisfies the following properties for each $s \in T$:
\begin{itemize}
\item $0$ is the least element and $1$ is the greatest element of $L$;
\item $[\psi\wedge\chi, s] = \min\{[\psi, s], [\chi, s]\}$
and $[\psi\vee\chi, s] = \max\{[\psi, s], [\chi, s]\}$;
\item $[\chi \imp \psi, s] = 1$ if $[\chi, s] \leq [\psi, s]$, and $[\chi \imp \psi, s] = [\psi, s]$ otherwise;
\item $[\chi \dimp \psi, s] = 0$ if $[\chi, s] \leq [\psi, s]$, and $[\chi \dimp \psi, s] = [\chi, s]$ otherwise;
\item $[\nx \psi, s] = [\psi, s+1]$;
\item $[\ps \psi, s] = \sup\{ [\psi, s+n] \mid n\in\mathbb{N}\}$ and
 $[\Box \psi, s] = \inf\{ [\psi, s+n] \mid n\in\mathbb{N}\}$.
\end{itemize}
These properties easily follow from the definitions. Again, let us verify the claims for $\Box$ and $\dimp$. 

For $\Box$, we first note that for any $n\in\mathbb{N}$ and any $w\in W$ such that $(w,s) \in \lb \Box\psi \rb$, we clearly must have $(w,s+n) \in \lb \psi\rb$, so that $[\Box\psi, s] \leq [\psi, s+n]$. Suppose now that $\chi$ and $s'$ are such that 
$\forall n\, \big([\chi, s'] \leq [\psi, s+n]\big).$
Then for any $w\in W$, if $(w, s') \in \lb\chi\rb$, then $(w, s+n) \in \lb \psi\rb$ for all $n$ and thus $(w, s) \in \lb \Box\psi\rb$ by definition, so that $[\chi, s'] \leq [\Box\psi, s]$ holds as needed.

For $\dimp$, there are two cases. Suppose first that $[\chi, s] \leq [\psi, s]$, so that whenever $(w,s) \in \lb\chi\rb$, we also have $(w,s) \in \lb\psi\rb$.
Then for any $w \in W$, we have 
\begin{align*}
(w, s) \in \lb \chi \dimp \psi \rb
&\leftrightarrow 
\exists v \geq w\, ((v, s) \in \lb\chi \rb \wedge (v,s)\not\in \lb\psi \rb ),
\end{align*}
which can never occur, so $[\chi\dimp \psi, s] = 0$.

Suppose, on the other hand, that $[\chi, s] > [\psi, s]$. Then there is $u \in W$ such that $(u, s) \in \lb\chi\rb$ but $(u,s) \not \in\lb\psi\rb$. 
Now, 
\begin{align*}
(w, s) \in \lb \chi \dimp \psi \rb
&\leftrightarrow 
\exists v \geq w\, ((v, s) \in \lb\chi \rb \wedge (v,s)\not\in \lb\psi \rb )\\
&\leftrightarrow (w,s) \in \lb\chi\rb,
\end{align*}
so that $[\chi\dimp\psi, s] = [\chi, s]$, as desired.

Now, $(L, \leq)$ is a countable linear order with endpoints, so it can be continuously embedded into the interval $[0,1]$ in such a way that the images of $0$ and $1$ are, respectively, $0$ and $1$. 
Let $\rho$ be such an embedding (we refer the reader to e.g., the proof of \cite[Theorem 5.1]{BPZ07} for an explicit construction of such a $\rho$.). We define a real valuation $V$ by setting $V(\psi, s) = \rho([\psi, s])$. By the properties above, $V$ is indeed a real valuation and $V(\varphi, t) = \rho([\varphi,t]) < 1$.
\end{proof}

\section{Labelled systems and quasimodels}\label{SecNDQ}

Our decidability proof for the G\"odel temporal validities $\gtl$ is based on (nondeterministic) quasimodels, originally introduced in \cite{Fernandez09} for \emph{dynamic topological logic,} a classical predecessor of \emph{intuitionistic temporal logic,} for which quasimodels were also used in \cite{F-D18}.
As the bi-relational semantics makes evident, G\"odel temporal logic is closely related to intuitionistic temporal logic. In this section we will introduce labelled spaces, labelled systems, and finally, quasimodels. Quasimodels can be viewed as a sort of nondeterministic generalisation of bi-relational models.

Of course, many decidability proofs for classical modal logics are obtained via the finite model property, so it is worthwhile to first note that this strategy cannot work for $\gtl$ because finite model properties do not hold. The finite model properties we define are of the form: \emph{falsifiable} implies \emph{falsifiable in a finite model}.
It is worth remarking that in sub-classical logics, it is indeed the notion of falsifiability that is relevant, as it is falsifiability that is dual to validity.
However, in view of our inclusion of coimplication, we may define ${\sim}\varphi\equiv \top\dimp\varphi$, and then it is not hard to check that ${\sim}\varphi$ is \emph{satisfiable} (in the sense of having a non-zero truth value) if and only if $\varphi$ is falsifiable.
Thus in view of the fact that our logic is (as we will see) {\sc pspace}-complete, validity, satisfiability, and falsifiability are all inter-reducible.

\begin{definition}\label{def:fmp}
The \define{strong finite model property} for $\gtl$ is the statement that if $\varphi \in \lanfull$ is falsifiable on a bi-relational model, then it is falsifiable on a bi-relational model $\mathcal  F=(W,T,{\leq},S, \lb\cdot\rb)$ where both $W$ and $T$ are finite.

The \define{order finite model property} for $\gtl$ is the statement that if $\varphi \in \lanfull$ is falsifiable on a bi-relational model, then it is falsifiable on a bi-relational model $\mathcal  F=(W,T,{\leq},S, \lb\cdot\rb)$ where $W$ is finite.

The \define{temporal finite model property} for $\gtl$ is the statement that if $\varphi \in \lanfull$ is falsifiable on a bi-relational model, then it is falsifiable on a bi-relational model $\mathcal  F=(W,T,{\leq},S, \lb\cdot\rb)$ where $T$ is finite.
\end{definition}

\begin{prop}\label{no_finite}
None of the finite model properties for $\gtl$ listed in \Cref{def:fmp} hold.
\end{prop}

\begin{proof}
Consider the formula $\ps (p \imp \nx p)$. To see that $\ps (p \imp \allowbreak\nx p)$ is falsifiable, let $-\omega $ denote $\ldots, -2,-1,0 $ with the usual ordering, and take the flow $(\omega, S)$, where $S$ is successor. Consider the model $(-\omega,\omega,{\leq},S, \lb\cdot\rb)$, where $\lb p\rb = \{(-n, m) \in -\omega \times \omega \mid n \geq m\}$. Then at each $(-i, i)$, the formula $p$ holds but $\nx p$ does not. Hence at each $(0, i)$, the formula $p \imp \nx p$ is falsified. Thus $\ps (p \imp \nx p)$ is falsified at $(0, 0)$. (See \Cref{fig:1} (left).)

\begin{figure}[h]

\begin{tikzpicture}[scale = 0.8]
\hspace{-.5cm}
\node[left] at (0,-0) {\phantom{$(w,t)$}};
\foreach \i in {3}
{
\draw (\i,0) -- (\i,-3);
\foreach \j in {\i,...,3}
{
\node[below left] at (\i,-\j) {$p$};
\draw[fill] (\i,-\j) circle (.05);
}
}
\foreach \i in {0, 1, 2}
{
\draw (\i,0) -- (\i,-3);
\foreach \j in {0,...,3}
{
\draw[dashed, ->] (\i,-\j) -- (\i+1, -\j);
}
\foreach \j in {\i,...,3}
{
\node[below left] at (\i,-\j) {$p$};
\draw[fill] (\i,-\j) circle (.05);
}
}
\foreach \i in {0,1,2}
\foreach \j in {0,...,\i}
\draw[fill=white] (\i+1,-\j) circle (.05);

\node at (3.5, -1.5) {\dots};
\node at (1.5, -3.5) {\vdots};
%\node at (3.45, -3.45) {$\ddots$};
\end{tikzpicture}
\hfill
\begin{tikzpicture}[scale = 0.8]
\hspace{-.3cm}
\foreach \i in {0, 1, 2}
{
\foreach \j in {0,1,2}
{
\draw (\i,-\j+0) -- (\i,-\j-.3);
\draw[decorate, decoration=snake] (\i,-\j-.3) -- (\i,-\j-.7);
\draw (\i,-\j-.7) -- (\i,-\j-1);
}
\foreach \j in {0,...,3}
{
\draw[dashed, ->] (\i,-\j) -- (\i+1, -\j);
}
}
\foreach \i in {3}
{
\foreach \j in {0,1,2}
{
\draw (\i,-\j+0) -- (\i,-\j-.3);
\draw[decorate, decoration=snake] (\i,-\j-.3) -- (\i,-\j-.7);
\draw (\i,-\j-.7) -- (\i,-\j-1);
}
}
\node at (3.5, -1.5) {\dots};
\node at (1.5, -3.5) {\vdots};

\node[left] at (0,-0) {$(w,t)$};
\node[below left] at (0,-1) {$p$};
\draw[fill] (0,-1) circle (.05);
\draw[fill=white] (1,-1) circle (.05);
\node[below left] at (1,-2) {$p$};
\draw[fill] (1,-2) circle (.05);
\draw[fill=white] (2,-2) circle (.05);
\node[below left] at (2,-3) {$p$};
\draw[fill] (2,-3) circle (.05);
\draw[fill=white] (3,-3) circle (.05);
%\node at (3.45, -3.45) {$\ddots$};
\end{tikzpicture}
\caption{Left: A bi-relational model falsifying $\ps (p \imp \nx p)$; right: $W$ and $T$ are necessarily infinite.}
\label{fig:1}%\caption{A bi-relational model falsifying $\ps (p \imp \nx p)$ (left); $W$ and $T$ are necessarily infinite (right).}
\end{figure}

%%%%THIS IS THE OLD VERSION

To see that $\ps (p \imp \nx p)$ can only be falsified on a model $(W,T,{\leq},S, \lb\cdot\rb)$ for which both $W$ and $T$ are infinite, suppose $\ps (p \imp \nx p)$ is falsified at $(w, t)$. Then there must be a sequence $w\geq w_0 > w_1> \dots$ such that for each $i \in \omega$, the formula $p$ holds on each $(w_i, S^i(t))$ but not on $(w_i, S^{i+1}(t))$. This clearly forces $W$ to be infinite, and by downward closure of $\lb\cdot\rb$, it forces $T$ to be infinite too. (See \Cref{fig:1} (right).)
\end{proof}

The same example as in Proposition \ref{no_finite} shows that under real semantics it is also the case that some formulas can only be falsified on an infinite flow with infinitely many realised truth values, as it forces $V(p,t ) >V(p ,S(t) )$ for all $t$.

Note that we have refuted all these finite model properties without using the $\dimp$ connective, thus in fact proving the stronger result that the finite model properties fail for the $\dimp$-free fragment.

We now introduce the structures we will use to mitigate the failure of these finite model properties.
Given a set $\Sigma\subseteq\lanfull$ that is closed under subformulas, we say that $\Phi\subseteq \Sigma$ is a \define{$\Sigma$-type} if the following occur.
\begin{enumerate}

\item $\bot\not\in \Sigma$.

\item If $\varphi\wedge\psi\in \Sigma$, then $\varphi\wedge\psi\in \Phi$ if and only if $\varphi,\psi\in \Phi$.

\item If $\varphi\vee\psi\in \Sigma$, then $\varphi\vee\psi\in \Phi$ if and only if $\varphi\in\Psi$ or $\psi\in \Phi$.

\item If $\varphi\imp \psi\in \Sigma$, then

\begin{enumerate}

\item  $\varphi\imp \psi\in\Phi$ implies that $\varphi\not \in \Phi$ or $\psi\in\Phi$,

\item $\psi \in \Phi$ implies that $\varphi\imp \psi \in\Phi$.

\end{enumerate}

\item If $\varphi\dimp \psi\in \Sigma$, then

\begin{enumerate}

\item $\varphi\dimp \psi\in\Phi$ implies $\varphi  \in \Phi$, 

\item  $\varphi\in\Phi$ and $\psi \notin \Phi$ implies that $\varphi\dimp \psi \in\Phi$.

\end{enumerate}

%\item If $\ps\psi\in \Sigma$ and $\psi\in \Phi$, then $\ps\psi\in \Phi$.

%\item If $\nec\psi\in \Phi$, then $ \psi \in \Phi$.

\end{enumerate}
The set of $\Sigma$-types will be denoted by $\type\Sigma$.
Often we want $\Sigma$ to be finite, in which case we write $\Sigma\Subset \lanfull$ to indicate that $\Sigma\subseteq \lanfull$ and $\Sigma$
is finite and closed under subformulas.
%For $\Phi\in\type\Sigma$, a \define{defect} of $\Phi$ is either a formula $\varphi\imp \psi\in\Sigma$ such that $\varphi \imp \psi\not\in\Phi$ and $\varphi\not\in \Phi$, or a formula $\varphi \dimp \psi \in \Phi$ such that $\psi \in \Phi$.
%The set of defects of $\Phi$ will be denoted $\partial_\Sigma\Phi$, or simply $\defect\Phi$ when $\Sigma$ is clear from context.

A partially ordered set $(A,{\leq})$ is \define{locally linear} %\david{I removed `strongly', since we don't use the weak version?}
 if it is a disjoint union of linear posets.%whenever $y \leq x$ and $z \leq x$, then either $y\leq z$ or $z\leq y$.

\begin{definition}\label{frame}
Let $\Sigma\subseteq\lanfull$ be closed under subformulas.
A \define{$\Sigma$-labelled space} is a triple $\mathcal  W= ( |\mathcal  W|,{\leq}_\mathcal  W ,\ell_\mathcal  W )$, where $( |\mathcal  W| ,{\leq}_\mathcal  W )$ is a locally linear poset and $\ell\colon |\mathcal  W| \to \type\Sigma$ an inversely monotone function (in the sense that $w\leq  v$ implies $\ell_\mathcal  W(w) \supseteq \ell_\mathcal  W(v)$) such that for all $w\in |\mathcal  W|$
\begin{itemize}
\item whenever $\varphi\imp \psi\in \Sigma \setminus \ell_\mathcal  W(w)$, there is $v\leq w$ such that $\varphi\in \ell_\mathcal  W(v)$ and $\psi \not \in \ell_\mathcal  W(v)$;
\item
whenever $\varphi\dimp \psi\in  \ell_\mathcal  W(w)$, there is $v\geq w$ such that $\varphi\in \ell_\mathcal  W(v)$ and $\psi\not\in \ell_\mathcal  W (v)$.
\end{itemize}
%Such a $v$ \define{revokes} $\varphi\imp \psi$.

The $\Sigma$-labelled space $\mathcal  W$ \define{falsifies} $\varphi\in\mathcal L$ if $\varphi\in \Sigma\setminus \ell_\mathcal  W(w)$ for some $w\in W$.
The \define{height} of $\mathcal W$ is the supremum of all $n$ such that there is a chain $w_1 <_\mathcal  W w_2 <_\mathcal  W \ldots <_\mathcal  W  w_n$.
%\brett{I don't think there was a plan for $\forall$ to be in signature; commented out.}%, and $\ell$ is \define{honest} if, for every $w\in |\mathcal  W|$ and every $\forall\varphi\in \Sigma$, we have that $\forall\varphi\in \ell(w)$ if and only if $\varphi\in \ell(v)$ for every $v\in |\mathcal  W|$.
\end{definition}

If $ \mathcal W$ is a labelled space, elements of $|\mathcal W|$ will sometimes be called \define{worlds.}
When clear from context we will omit subscripts and write, for example,~$\leq$ instead of $\leq_\mathcal  W$.

A \define{convex relation} %\david{Removed `pointwise-'}
 between posets $(A,{\leq}_A)$ and $(B,{\leq}_B)$ is a binary relation $R\subseteq A\times B$ such that for each $x \in A$ the image set $\{y \in B \mid x \mathrel R y\}$ is convex with respect to $\leq_B$, and for each $y \in B$ the preimage set $\{x \in A \mid x \mathrel R y\}$ is convex with respect to $\leq_A$.
%\define{weakly order-preserving relation} between posets $(A,\leq_A)$ and $(B,\leq_B)$ is a binary relation $R\subseteq A\times B$ such that if $x R y$, $x' R y'$, and $x < x'$, then $y \leq y'$.
 The relation $R$ is \define{fully confluent} if it validates the four following conditions:
\begin{description}

\item[forth-down]\label{forward-down} if $x \leq _A x' \mathrel R y'$ there is $y$ such that $x \mathrel R y \leq_B y'$,

\item[forth-up]\label{forward-up} if $x' \geq _A x \mathrel R y$ there is $y'$ such that $x' \mathrel R y' \geq_B y$, 

\item[back-down]\label{backward-down} if $x' \mathrel R y' \geq_B y$ there is $x$ such that $x'  \geq _A x \mathrel R y$,

\item[back-up] if $x \mathrel R y \leq_B y'$ there is $x'$ such that $x \leq _A x' \mathrel R y'$.

\end{description}
%or more concisely, if ${\leq_A} \circ R = R \circ {\leq_B}$ and ${\geq_A} \circ R = R \circ {\geq_B}$.

\begin{definition}\label{compatible}
Let $\Sigma\subseteq\lanfull$ be closed under subformulas. Suppose that $\Phi,\Psi\in\type\Sigma$. The ordered pair $(\Phi,\Psi)$ is \define{sensible} if
\begin{enumerate}
\item for all $\nx\varphi\in \Sigma$, we have $\nx\varphi\in \Phi$ if and only if $ \varphi\in \Psi$,
\item for all $\ps\varphi\in \Sigma$, we have $\ps\varphi\in \Phi$ if and only if $\varphi\in\Phi$ or $\ps\varphi\in \Psi$, 
\item for all $\nec \varphi\in \Sigma$, we have $\nec \varphi\in \Phi$ if and only if $\varphi\in \Phi$ and $\nec \varphi\in \Psi$.
\end{enumerate}
A pair $(w,v)$ of worlds in a labelled space $\mathcal  W$ is \define{sensible} if $(\ell (w),\ell (v))$ is sensible.
A relation
$S\subseteq |\mathcal  W|\times |\mathcal  W|$
is \define{sensible} if every pair in $S$ is sensible.
Further, $S$ is \define{$\omega$-sensible} if it is serial and
\begin{itemize}

\item whenever $\ps\varphi\in \ell(w)$, there are $n\geq 0$ and $v$ such that $w \mathrel S^n v$ and $\varphi\in \ell(v)$; 

\item whenever $\nec\varphi\notin \ell(w)$, there are $n\geq 0$ and $v$ such that $w \mathrel S^n v$ and $\varphi\notin \ell(v)$.

\end{itemize}

A \define{labelled system} is a labelled space $\mathcal  W$ equipped with a fully confluent convex sensible relation $R_\mathcal W\subseteq |\mathcal  W|\times|\mathcal  W|$. If moreover $R_\mathcal W$ is $\omega$-sensible, we say that $\mathcal  W$ is a \define{$\Sigma$-quasimodel.}

\end{definition}

Any bi-relational model can be regarded as a $\Sigma$-quasimodel. If $\mathcal {X} = (W,T,{\leq},\allowbreak S, \val \cdot)$ is a bi-relational model and $x\in W \times T$, we can assign a $\Sigma$-type $\ell_\mathcal  X(x)$ to $x$ given by
$\ell_\mathcal  X(x)=\cbra\psi\in \Sigma \mid x\in \val\psi \cket.$
We also set $R_\mathcal  X=\{((w, t) , (w, S(t))) \mid w \in W, t \in T\}$; it is obvious that $R_\mathcal  X$ is $\omega$-sensible.
Henceforth we will tacitly identify $\mathcal  X$ by its associated $\Sigma$-quasimodel.

\section{From quasimodels to bi-relational models}\label{SecGen}

%\david{This section needs to be entirely rewritten (Brett). However, there are some useful notions here, so I'll leave it.}

If $\mathcal  Q$ is a quasimodel whose $\omega$-sensible relation is $R_\mathcal Q$, then we may not view $\mathcal  Q$ directly as a bi-relational model for the primary reason that $R_\mathcal  Q$ is not necessarily a function. However, we can extract bi-relational models from quasimodels via an unwinding construction. More precisely, given $\Sigma \Subset \lanfull$, suppose $\varphi$ is falsified on the $\Sigma$-quasimodel $\mathcal Q$. In this section we show how to obtain from $\mathcal Q$ a bi-relational model $\lm_\varphi$ satisfying $\varphi$. We call the resulting bi-relational model a \emph{limit model} of $\mathcal Q$. This proves $\gtl$ is sound for the class of quasimodels.

 The general idea for determinising $\mathcal Q$ is to consider infinite  paths on $\mathcal  Q$ as points in the limit model. However, we will only select paths $\vec w$ with the property that, if $\ps\varphi$ occurs in $\vec w$, then $\varphi$ must also occur at a later time, with a dual condition for $\nec$. These are the realising paths of $\mathcal  Q$.

\begin{definition}
A \define{path} in a $\Sigma$-quasimodel $\mathcal  Q$ is any  sequence $\<w_i\>_{i < \alpha}$ with $\alpha \leq \omega$ such that $w_i \mathrel R w_{i+1}$ whenever $i+1 < \alpha$. An infinite path
${\vec w}=\< w_i\>_{i<\omega}$
is \define{realising} if for all $i<\omega$
\begin{itemize}
\item
 for all $\ps\psi\in \ell(w_i)$, there exists $j\geq i$ such that $\psi\in \ell  (w_j)$,
\item
 for all $\nec\psi\in \Sigma \setminus \ell(w_i)$, there exists $j\geq i$ such that $\psi\in \Sigma \setminus \ell  (w_j)$.
\end{itemize}
\end{definition}

Denote the set of realising paths by $ \domlm $, and let $\< v_i\>_{i<\omega}\leq \< w_i\>_{i<\omega}$ if and only if $ v_i\leq w_i $ %\david{Changed from $\ell_\mathcal  Q(v_i) \supseteq \ell_\mathcal  Q(w_i)$.}
 for all $i < \omega$.
The worlds of the limit model will be a linearly ordered (with respect to $\leq$) subset of $\lm$. 
 Let
%${\vec w}=( w_i)_{i<\omega}$
%be a realising path and $\basic n{\Slm ({\vec w})}$ be a neighborhood of $\Slm({\vec w})$. Then if
%${\vec v}\in \basic{n+1}{\vec w},$
%$w_i\acc v_i$ for all $i\leq n+1$, so $w_{i+1}\acc v_{i+1}$ for all $i\leq n$ and
%$\Slm(\vec v)\in \basic n{\Slm(\vec w)}.$
%Hence
%$\Slm (\basic{n+1}{\vec w})\subseteq \basic n{\Slm(\vec w)},$
%and $\Slm$ is continuous.\endproof

Given our $\Sigma \Subset \lanfull$ and a formula $\varphi$ falsified in $\mathcal Q$, the limit model $\lm_\varphi$ will be of the form $(W, \omega, {\leq}, S, \allowbreak\val \cdot)$, where $W \subseteq \domlm$, the flow function $S \colon \omega \to \omega$ is successor, and $\val p = \{(\vec w ,i) \mid p \in \ell(w_i)\}$ (extended to compound formulas in accordance with \Cref{DefKSem}). We now describe how to select the linearly ordered subset of realising paths $W$.

\begin{definition}
A \define{finite grid} is a linearly ordered finite set of paths in $\mathcal Q$  of some uniform length $k < \omega$ (with the pointwise ordering). A finite grid $P'$ with paths of length $k'$, \define{extends} a finite grid $P$ with paths of length $k$, if $k' > k$ and $P \subseteq P'|_k$, where $P'|_k$ is the restriction of the paths in $P'$ to their initial $k$-length segments.
\end{definition}

\begin{definition}[defects]
 Let $P$ be a finite grid.
\begin{itemize}
\item
A \define{$\ps$-defect} of $P$ is a pair $((w_i)_{i<k}, \ps\varphi) \in P \times \Sigma$ such that $\ps\varphi \in \ell(w_{k-1})$, but $\varphi \not\in \ell(w_{k-1})$.
\item
A \define{$\nec$-defect} of $P$ is a pair $((w_i)_{i<k}, \nec\varphi) \in P \times \Sigma$ such that $\nec\varphi \not\in \ell(w_{k-1})$, but $\varphi \in \ell(w_{k-1})$.
\item
An \define{$\imp$-defect} of $P$ is a triple $((w_i)_{i<k}, \varphi \imp\psi, j ) \in P \times \Sigma \times k$ such that $\varphi \imp\psi \not\in \ell(w_{j})$, but there is no $(v_i)_{i<k} \leq (w_i)_{i<k}$ also in $P$ such that $\varphi \in \ell(v_{j})$, but $\psi \not\in \ell(v_{j})$.
\item
A \define{$\dimp$-defect} of $P$ is a triple $((w_i)_{i<k}, \varphi \dimp\psi, j ) \in P \times \Sigma \times k$ such that $\varphi \dimp\psi \in \ell(w_{j})$, but there is no $(v_i)_{i<k} \geq (w_i)_{i<k}$ also in $P$ such that $\varphi \in \ell(v_{j})$, but $\psi \not\in \ell(v_{j})$.
\item
A \define{seriality defect} is a path $(w_i)_{i<k} \in P$.
\end{itemize}
\end{definition}

Note that because $\Sigma$ is finite, any finite grid has a finite number of defects. We select the set $W \subseteq \domlm$ as follows. We maintain a first-in-first-out queue $D$ of defects and a finite grid $P$ that is extended each time we process a defect. (We will ensure that the constituents of $D$ continue to be defects of the grid.) Choose some $w \in \domlm$ satisfying $\varphi$. We initialise the grid to the single sequence $(w)$ (of length 1) and add all defects of this grid to $D$.

Now, inductively, consider the defect at the head of the queue $D$. (Because of seriality defects, $D$ can never be empty.) 
\begin{itemize}
\item
If the defect is a $\ps$-defect $((w_i)_{i<k}, \ps\varphi)$ of $P$, because $R=R_\mathcal  Q$ is $\omega$-sensible,  we know that there exist $j> 0$ and $v\in R^{j}(w_{k-1})$ such that $\varphi\in \ell(v)$.  We then define $k'=k+j$ and choose
$w_{k -1}\mathrel R w_{k}\mathrel R \hdots \mathrel R  w_{k'-1}=v.$
By the forth-up confluence property, we can extend every $ (v_i)_{i < k} > (w_i)_{i<k}$ in $P$ to a $k'$-length sequence in a way that preserves the (linear) ordering on sequences; similarly for every $ (v_i)_{i < k} < (w_i)_{i<k}$ in $P$ using the forth-down confluence property. 
\item
If the defect is a $\nec$-defect $((w_i)_{i<k}, \nec\varphi)$, we can find $j> 0$ and $v\in R^{j}(w_{k-1})$ such that $\varphi\not\in \ell(v)$, and proceed as for the $\ps$ case.

\item
If the defect is an $\imp$-defect $((w_i)_{i<k}, \varphi \imp\psi, j )$, choose $v_j < w_j$ such that $\varphi \in \ell(v_j)$ and $\psi \not \in \ell(v_j)$. %None of the sequences in $P$ below $(w_i)_{i<k}$ have a $j$th element with the same label as $v_j$ else $((w_i)_{i<k}, \varphi \imp\psi, j )$ would not be a defect. Neither do any of the sequences above $(w_i)_{i<k}$, since $\ell$ is inversly monotone and $\varphi \in \ell(v_j)$ but $\varphi \not \in \ell(w_j)$. Hence $\ell(v_j)$ is a new label for the $j$th position, and we can find a  
Let $(u_i)_{i<k}$ be the minimum sequence in $P$ with $u_j > v_j$ and $(t_i)_{i<k}$ be the maximum sequence in $P$ with $t_j < v_j$, if it exists.
We will assume that $(t_i)_{i<k}$ is defined, since the case where $v_j$ is not bounded below is similar but simpler.
We will complete the sequence $(v_i)_{i<k}$ so that $(t_i)_{i<k} < (v_i)_{i<k} < (u_i)_{i<k}$. Suppose we have defined $v_{j'}$ for $j\geq j' < k-1$. To define $v_{j'+1}$ choose $y$ with $v_{j'} \mathrel R y$ and $y \leq u_{j'+1}$, which exists by forth-up confluence. If $y \geq t_{j'+1}$ we can set $v_{j'+1} = y$ and we are done. Otherwise, by strong local linearity of $\leq$, we have $y < t_{j'+1}$. In this case, choose $z$ with $v_{j'} \mathrel R z$ and $z \geq t_{j'+1}$, which exists by forth-down confluence. Then $y < t_{j'+1} \leq z$, so as the image set of $v_{j'}$ under $R$ is convex (by convexity of $R$), we have $v_{j'} \mathrel R t_{j'+1}$ and we can set $v_{j'+1} = t_{j'+1}$. In this way we can define $(v_i)_{i<k}$ inductively for all indices greater than $j$. The process for indices less than $j$ is similar, using back-up and back-down confluence and the convexity of preimage sets under $R$. By construction, $(v_i)_{i<k}$ sits strictly between  $(t_i)_{i<k}$ and $(u_i)_{i<k}$.

\item 

The case for a $\dimp$-defect is the same, except that we choose $v_j>w_j$.

\item
For a seriality defect $(w_i)_{i<k}$, we extend $(w_i)_{i<k}$ to $(w_i)_{i<k+1}$,  then extend the other sequences to length $k+1$ using  forth-up/forth-down confluence.
\end{itemize}
We have finished updating our grid, which we temporarily denote $P'$; now we must update $D$ so that it contains defects of $P'$.
First we remove the head of $D$ %\david{Probably call this $D'$.}
 (the defect we just processed). Then every $((w_i)_{i<k}, \ps\varphi) \in D$ is rewritten as $((w_i)_{i<k'}, \ps\varphi)$, every $((w_i)_{i<k}, \nec\varphi) \in D$ as $((w_i)_{i<k'}, \nec\varphi)$, every  $((w_i)_{i<k}, \varphi \imp\psi, j ) \in D$ as $((w_i)_{i<k'}, \varphi \imp\psi, j )$, and every $(w_i)_{i<k}$ as $(w_i)_{i<k'}$. %(This for every $(w_i)_{i<k'}$ in $P'$, recalling that $P'$ may contain multiple extensions of $(w_i)_{i<k}$.)
 Next, all elements of $D$ that are not defects of $P'$ are deleted from $D$. Finally, all $\ps$-, $\nec$-, and $\imp$-defects of $P'$ that are not already in $D$, and all seriality defects of $P'$ are added to the back of the queue.
%Finally, we will use $\ell$ to define a valuation: if $p$ is a propositional variable, set
%\[\val p^\ell= \big \{ {\vec w}\in \domlm : p \in \ell(w_0) \big \}.\]
%With this, we are now ready to assign a dynamic topological model to each quasimodel:

%\begin{definition}
%Given a $\Sigma$-quasimodel $\mathcal {Q},$ define
%\[\vec {\mathcal {Q}}=\Big ( \domlm,\Tlm,\Slm,\lb\cdot\rb^\ell \Big ) \]
%to be the \define{limit model} of $\mathcal {Q}$.
%\end{definition}

We set $W$ to be the limit of this sequence of finite grids. More precisely, let the sequence of grids be $P_0, P_1,\dots$, containing paths of length $k_0, k_1\dots$, respectively; then $W$ is the set of paths $(w_i)_{i<\omega}$ for which there exists $N$ such that for all $n>N$ the initial segment $(w_i)_{i<k_n}$ is in $P_n$. (This gives us the limit we would expect and forbids `diagonal' sequences whose initial segments appear in $P_0, P_1,\dots$ only infinitely often.)

By construction, $(W, \leq)$ is a linearly ordered set, and for each $p \in \mathbb P$, the set $\lb p \rb$ is downward closed in its first coordinate. Thus the limit model $\lm_\varphi = (W, \omega, \leq, S, \val \cdot)$ indeed defines a bi-relational model. Of course $\lm_\varphi$ is only useful if $\lb\cdot\rb$ `matches' with $\ell$ on all formulas of $\Sigma$, not just propositional variables. Fortunately, this turns out to be the case.

\begin{lem}\label{sound}
Let $\Sigma\Subset\lanfull$, $\mathcal  Q$ be a $\Sigma$-quasimodel, $\varphi, \psi\in \Sigma$, and $\lm_\varphi = (W, \omega, \leq, S, \val \cdot)$ be as described above.
Then
$\lb\psi\rb=\{(\vec w, i) \mid \psi\in \ell(w_i)\}.$
\end{lem}
\begin{proof} The proof goes by standard induction of formulas. The induction steps for $\wedge$, $\vee$, and $\bot$ are immediate. The cases for $\imp,\nx,\ps$, and $\nec$ follow straightforwardly from the construction of $W$, because every defect is eventually eliminated.
\end{proof}

We obtain the main result of this section, which in particular implies that $\gtl$ is sound for the class of quasimodels.

\begin{prop}\label{second}
Let $\Sigma\Subset\lanfull$ and $\mathcal Q$ be any $\Sigma$-quasimodel, and suppose $\varphi$ is falsified on $\mathcal Q$. Then there exists  a bi-relational model $\lm_\varphi$ that falsifies $\varphi$.
\end{prop}

It is interesting to note that although we assumed in this section that $\Sigma$ was finite, this restriction can be removed. Since $\lanfull$ is countable, for an arbitrary subformula-closed $\Sigma \subseteq \lanfull$, there can only be a countable number of defects in any finite grid. Thus with appropriate scheduling all defects can be eliminated in the limit.

%It is clear from its definition that $\Slm$ is a function and by Lemma \ref{lemSisCont} it is continuous, so that $(\domlm,\Tlm,\Slm)$ is a dynamical system. Now, suppose that $w_0\in |\mathcal  Q|$ is such that $\varphi \in \Sigma \setminus \ell(w_0)$. By Lemma \ref{extension}, $w_0$ can be extended to a realising path ${\vec w}$. It follows from Lemma \ref{sound} that ${\vec w}\not\in \val\varphi^\ell,$ and therefore $\varphi$ is falsified on $\lm$. Conversely, if $\vec w = (w_i)_{i<\omega} \in \domlm$ falsifies $\varphi$ then $\varphi \not \in \ell(w_0)$, so that $\fr Q$ falsifies $\varphi$ if and only if $\lm$ does.
%\endproof

\color{black}

\section{From bi-relational models to finite quasimodels}\label{Sec:quotient}

%\brett{As it stands, linearising the bi-relational model is not needed, because \Cref{DefKSem} says they are linear by definition.}% Here is a proof idea. Let $(W,T,\leq,S,\val\cdot)$ be a bi-relational model, and let $B $ be the maximal bisimulation between $W\times T$ and $M_\Sigma$. Then $B$ is a bisimulation. The claim is then that $B(W\times T)$, as a substructure of $\moments\Sigma$, is a $\Sigma$-quasimodel.

As we noted earlier, every bi-relational model can be naturally viewed as a quasimodel. However, we wish to show that, given a finite and subformula-closed $\Sigma$, we can from each bi-relational model $\mathcal X$  produce a \emph{finite} quasimodel satisfying exactly the same formulas from $\Sigma$ as $\mathcal X$. In this section we do just this by transforming $\mathcal X$ in two steps. First, we will take a bisimulation quotient to obtain a finite $\Sigma$-labelled space equipped with a fully confluent $\omega$-sensible relation. The second step will be to extend the $\omega$-sensible relation to be convex, yielding a finite quasimodel.

We describe the quotient explicitly, noting afterwards that it is a particular type of bisimulation quotient.

Let $\Sigma$ be a subformula-closed subset of $\lanfull$, and let $\mathcal {X} = (W,T,{\leq},S, \val \cdot)$ be a bi-relational model. For $x\in W \times T$, define $\ell_\mathcal X(x)$ by
$\ell_\mathcal X(x)=\cbra\psi\in \Sigma \mid x\in \val\psi \cket$, and define $L_\mathcal X(x) = \cbra \ell(y) \mid \pi_2(x) = \pi_2(y) \cket$, where $\pi_2 \colon W \times T \to T$ is the projection $(w, t) \mapsto t$. We define the binary relation $\sim$ on $W \times T$ by \[x \sim y \iff (\ell(x), L(x)) = (\ell(y), L(y)).\]
If $\Sigma$ is finite, then clearly $(W \times T) / {\sim}$ is finite.

Note that $\sim$ is the largest relation that is simultaneously a bisimulation with respect to the relations $\leq$ and $\geq$, with $\Sigma$ treated as the set of atomic propositions that bisimilar worlds must agree on.

Now define a partial order $\leq_\mathcal Q$ on the equivalence classes $(W \times T) /{\sim}$ of $\sim$ by
\[[x] \leq_\mathcal Q [y] \iff L(x) = L(y)\text{ and }\ell(x) \supseteq \ell(y),\]
noting that this is well-defined and is indeed a partial order.

Since each set $L(x)$ can be linearly ordered by inclusion and $\ell(x) \in L(x)$, the poset $((W \times T) / {\sim}, \leq_\mathcal Q)$ is a disjoint union of linear orders. By defining $\ell_Q$ by 
$\ell_\mathcal Q([x]) = \ell(x)$
we obtain a $\Sigma$-labelled space $((W \times T) / {\sim}, \leq_\mathcal Q, \ell_\mathcal Q)$; it is not hard to check that this labelling is inversely monotone and that the clauses for $\imp$ and $\dimp$ hold with this labelling.

Now define the binary relation $R_\mathcal Q$ on $(W \times T) / {\sim}$ to be the smallest relation such that $[(w, t)] \mathrel R_\mathcal Q [(w, S(t))]$, for all $(w, t) \in W \times T$. 

\begin{lem}
The relation $R_\mathcal Q$ is fully confluent and $\omega$-sensible.
\end{lem}

\begin{proof}
For confluence, suppose $[(w, t)] \mathrel R_\mathcal Q [(w, S(t))]$. To see that the forth-up condition holds, suppose further that $[(w, t)] \leq_\mathcal Q [y]$. Then as $l(y) \in L(y) = L(x)$ there is some $v \geq w$ with $[y] = [(v, t)]$. Then we have $[(v, t)] \mathrel R_\mathcal Q [(v, S(t))]$ and $[(w, S(t))] \leq_\mathcal Q [(v, S(t))]$, as required for the forth-up condition. The proofs of the remaining three confluence conditions are entirely analogous.

It is clear that $R_\mathcal Q$ is sensible. To see that $R_\mathcal Q$ is $\omega$-sensible, first note that $R_\mathcal Q$ is clearly serial. Next, suppose that $\ps \varphi \in \ell_\mathcal Q([(w, t)])$. Then by our definitions, $(w, t) \in \val {\ps\varphi}$. Thus $(w, S^n(t)) \in \val \varphi$ for some $n > 0$. It follows that $[(w, t)] \mathrel R_\mathcal Q [(w, S(t))] R_\mathcal Q \dots R_\mathcal Q  [(w, S^n(t))]$. Similar reasoning applies when we suppose that $\nec \varphi \not\in \ell_\mathcal Q([(w, t)])$, completing the proof.
\end{proof}

As promised, we now have a $\Sigma$-labelled space equipped with a fully confluent $\omega$-sensible relation. We now transform this labelled space into a $\Sigma$-quasimodel by making the additional relation convex by fiat.

Define $R^+_\mathcal Q$ by $X \mathrel R^+_\mathcal Q Y$ if and only if there exist $X_1 \leq_\mathcal Q X \leq_\mathcal Q X_2$ and $Y_1 \leq_\mathcal Q Y \leq_\mathcal Q Y_2$ such that $X_2 \mathrel R_\mathcal Q Y_1$ and $X_1 \mathrel R_\mathcal Q Y_2$. Now define $\mathcal Q = ((W \times T) / {\sim}, \leq_\mathcal Q, \ell_\mathcal Q, R^+_\mathcal Q)$.

\begin{lem}
The structure $\mathcal Q$ is a $\Sigma$-quasimodel.
\end{lem}

\begin{proof}
We already know that $((W \times T) / {\sim}, \leq_\mathcal Q, \ell_\mathcal Q)$ is a $\Sigma$-labelled space. First we must check $R^+_\mathcal Q$ is still fully confluent and $\omega$-sensible. 

For the forth-down condition, suppose $X \leq_\mathcal Q X' \mathrel R^+_\mathcal Q Y'$. Then by the definition of $R^+_\mathcal Q$, there are some $X_2 \geq_\mathcal Q X'$ and $Y_1 \leq_\mathcal Q Y'$ such that $X_2 \mathrel R_\mathcal Q Y_1$. %Then we also know that $X_2 = [(w, t)]$ and $Y_1 = [(w, S(t))]$ for some $w, t$.
 Since $X \leq_\mathcal Q X' \leq_\mathcal Q X_2$, by the forth-down condition for $R_\mathcal Q$ there is some $Y \leq_\mathcal Q Y_1$ with $X \mathrel R_\mathcal Q Y$ and therefore $X \mathrel R^+_\mathcal Q Y$. Since $Y \leq_\mathcal Q Y_1 \leq_\mathcal Q Y'$, we are done. % it follows that $X$ can be written in the form $[(v, t)]$, for some $v \leq w$. Then $[(v, t)] \mathrel R_\mathcal Q [(v, S(t)]$, and thus $[(v, t)] \mathrel R^+_\mathcal Q [(v, S(t)] \leq_\mathcal Q [(w, S(t)] \leq_\mathcal Q Y'$. That is, $[(v, S(t))]$ witnesses the satisfaction of the forth-down condition.
 The proof that the forth-up condition holds is just the order dual of that for forth-down. The proofs of the back-down and back-up conditions are similar.

%For the back-down condition, suppose $ X' \mathrel R^+_\mathcal Q Y' \geq_\mathcal Q Y$. Then by the definition of $R^+_\mathcal Q$, there are some $X_1 \leq_\mathcal Q X'$ and $Y_2 \geq_\mathcal Q Y'$ such that $X_1 \mathrel R_\mathcal Q Y_2$. Then we also know that $X_1 = [(w, t)]$ and $Y_2 = [(w, S(t)]$, for some $w, t$. can be written in the form $[(v, S(t))]$ for some $v \leq w$. Then $[(v, t)]$ witnesses the satisfaction of the back-down condition. 

To see that $R^+_\mathcal Q$ is sensible, suppose $ X \mathrel R^+_\mathcal Q Y$ and that $\nx \varphi \in \Sigma$. Take $X_1 \leq_\mathcal Q X \leq_\mathcal Q X_2$ and $Y_1 \leq_\mathcal Q Y \leq_\mathcal Q Y_2$ such that $X \mathrel R_\mathcal Q Y_1$. Then
\begin{align*}
\nx \varphi \in \ell_\mathcal Q(X) &\implies \nx \varphi \in \ell_\mathcal Q(X_1)
&&\implies \phantom{\nx}\varphi \in \ell_\mathcal Q(Y_2)
 &&\implies \phantom{\nx}\varphi \in \ell_\mathcal Q(Y)\\
 &\implies \phantom{\nx}\varphi \in \ell_\mathcal Q(Y_1)
 &&\implies \nx \varphi \in \ell_\mathcal Q(X_2)  &&\implies \nx \varphi \in \ell_\mathcal Q(X),
\end{align*}
so $ \nx \varphi \in \ell_\mathcal Q(X) \iff  \varphi \in \ell_\mathcal Q(Y)$. The $\ps$ and $\nec$ cases are similar. It is now clear that $R^+_\mathcal Q$ is $\omega$-sensible since $R_\mathcal Q \subseteq R^+_\mathcal Q$, and the three conditions for a sensible relation to be $\omega$-sensible are all monotone.

Finally, we show that $R^+_\mathcal Q$ is convex. Firstly, for the image condition, if $X \mathrel R^+_\mathcal Q Y_1$ and $X \mathrel R^+_\mathcal Q Y_2$ with $Y_1 \leq_\mathcal Q Y \leq Y_2$, then by the definition of $R^+_\mathcal Q$ we can find  $X_2 \geq_\mathcal Q X$ and $Y'_1 \leq_\mathcal Q Y_1$ with $X_2 \mathrel R_\mathcal Q Y'_1$, and similarly $X_1 \leq_\mathcal Q X$ and $Y'_2 \geq_\mathcal Q Y_2$ with $X_1 \mathrel R_\mathcal Q Y'_2$. Since then $X_1 \leq_\mathcal Q X \leq_\mathcal Q X_2$ and $Y'_1 \leq_\mathcal Q Y \leq_\mathcal Q Y'_2$, by the definition of $R^+_\mathcal Q$ we conclude that $X \mathrel R^+_\mathcal Q Y$. The preimage condition is completely analogous. This completes the proof that $\mathcal Q$ is a $\Sigma$-quasimodel.
\end{proof}

\begin{lem}\label{falsifies}
Let $\varphi \in \Sigma$. Then $\mathcal X$ falsifies $\varphi$ if and only if $\mathcal Q$ falsifies $\varphi$.
\end{lem}

\begin{proof}
We have: $\mathcal X$ falsifies $\varphi$ if and only if $\exists (w, t) \in (W \times T) \setminus \lb\varphi\rb$ if and only if $\exists [(w, t)] \in (W \times T) / {\sim}$ with $\varphi \in \Sigma \setminus \ell_\mathcal Q([(w, t)])$ if and only if $\mathcal Q$ falsifies $\varphi$.
\end{proof}

In order to use Lemma~\ref{falsifies} to prove decidability, we need to compute a bound on the size of the quasimodel $\mathcal Q$ in terms of the size of $\Sigma$, when $\Sigma$ is finite.

\begin{lem}\label{quasi:bound}
Suppose $\Sigma$ is finite, and write $\lgt\Sigma$ for its cardinality. Then the height of $\mathcal Q$ is bounded by $\lgt\Sigma +1$, and the cardinality of the domain $(W \times T) / {\sim}$ of $\mathcal Q$ is bounded by $(\lgt\Sigma +1)\cdot 2^{\lgt \Sigma (\lgt \Sigma +1)+1}$
\end{lem}

\begin{proof}
Each element of the domain of $\mathcal Q$ is a pair $(\ell, L)$ where $L$ is a (nonempty) subset of $\raisebox{2pt}{$\wp$} \Sigma$ and $\ell \in L$. Since $L$ is linearly ordered by inclusion, it has height at most $\lgt \Sigma + 1$. There are $(2^{\lgt \Sigma})^i$ subsets of $\raisebox{2pt}{$\wp$} \Sigma$ of size $i$, so there are at most $\sum_{i = 1}^{\lgt \Sigma +1}(2^{\lgt \Sigma})^i$ distinct $L$. The sum is bounded by $2^{\lgt \Sigma (\lgt \Sigma +1)+1}$. The factor of $\lgt\Sigma +1$ corresponds to choice of an $\ell \in L$, for each $L$.
\end{proof}

Thus we have an exponential bound on the size of $\mathcal Q$. Now we can prove that $\gtl_{\mathbb R}$ and $\gtl$ are decidable.

\begin{thm}
The logic $\gtl_{\mathbb R}$ of $\lanfull$-formulas that are valid on all flows and the logic $\gtl$ of $\lanfull$-formulas that are valid on all bi-relational frames are equal and decidable.
\end{thm}

\begin{proof}
By \Cref{sec:real}, $\gtl_{\mathbb R} = \gtl$. Since falsifiability is the complement of validity, it suffices to show that it is decidable whether a formula $\varphi$ is falsifiable over the class of all bi-relational frames. Let $\Sigma$ be the set of subformulas of $\varphi$. If $\varphi$ is falsifiable in a $\Sigma$-quasimodel of size at most $(\lgt\Sigma +1)\cdot 2^{\lgt \Sigma (\lgt \Sigma +1)+1}$, then by Proposition~\ref{second}, $\varphi$ is falsified in a bi-relational frame. Conversely, if $\varphi$ is falsified in a bi-relational frame, then by Lemma~\ref{falsifies} and Lemma~\ref{quasi:bound}, $\varphi$ is falsified in a $\Sigma$-quasimodel of size at most $(\lgt\Sigma +1)\cdot 2^{\lgt \Sigma (\lgt \Sigma +1)+1}$.
Hence, it suffices to check falsifiability of $\varphi$ on the set of all $\Sigma$-quasimodels of size at most $(\lgt\Sigma +1)\cdot 2^{\lgt \Sigma (\lgt \Sigma +1)+1}$. It is clear that this check can be carried out within a computable time bound; hence the problem is decidable.
\end{proof}

Note that this proof yields only a {\sc nexptime} upper bound.
In the next section, we will see that this can be improved.

\section{PSPACE completeness}\label{Sec:PSPACE}%\brett{Removed formatting on {\sc PSPACE}, which is inadvisable in section headings for the same reasons as in paper titles (also I find it slightly ugly)}

 We recall that the validity problem for {\sf LTL} is {\sc pspace}-com\-plete~\cite[Theorem 4.1]{SK85}. Thus to prove {\sc pspace}-hardness of the $\gtl$ validity problem, it suffices to give a reduction from {\sf LTL} validity to $\gtl$ validity.
  
 Consider the (negative) translation ``$(\cdot)^{\bullet}$''~\cite{91986055} defined as follows:
 
 \begin{enumerate}
 	\item $p^{\bullet} = \neg \neg p$, with $p$ a propositional variable;
 	%	\item $(\varphi \wedge \psi)^{N} = (\psi)^{N} \wedge (\psi)^{N}$;
 	%	\item $(\varphi \rightarrow \psi)^{N} = (\psi)^{N} \rightarrow (\psi)^{N}$;
 	%	\item $(\varphi \vee \psi)^{N} =  (\psi)^{N} \vee (\psi)^{N}$;
 	%	\item $(\nec \psi)^{N} = \nec (\psi)^{N}$;
 	%	\item $(\ps \psi)^{N} = \neg \nec \neg (\psi)^{N}$;
 	\item Homomorphic for the rest of operators
 \end{enumerate}

In what follows we may assume that $T=\mathbb N$, equipped with the standard successor function.
 
 \begin{prop}
 Given any real G\"odel temporal valuation $V$, any formula $\varphi \in \mathcal L$ and any $t\in \mathbb{N}$, we have that $V(\varphi^\bullet,t) \in \lbrace 0,1 \rbrace$.
 \end{prop}
 \begin{proof}
 	By structural induction.
 	\begin{description}
 		\item[\sc Case $\varphi\in\mathbb P$:] This case follows from observing that for any $\psi \in \mathcal L$, we have  $V(\neg \psi,t)\allowbreak\in \{0,1\}$ regardless of $V(\psi,t)$, by the definition of negation.
 		\item[\sc Case $\varphi=\psi\odot\theta$,  $\odot\in \{\wedge,\vee\}$:] 
By induction, $V(\psi^\bullet,t), V(\theta^\bullet,t)\allowbreak\in \{0,1\}$; hence their maximum and minimum are also elements of $\{0,1\}$. 		
 			
 		\item[\sc Case $\varphi=\psi\imp\theta$:]
 		 		By definition, $V(\psi^\bullet \imp\theta^\bullet,t)$ is either $1$ or $V(\theta^\bullet,t)$, which by the induction hypothesis is an element of $\{0,1\}$.
 		\item[\sc Case $\varphi=\psi\dimp\theta$:]
 				By definition, $V(\psi^\bullet\dimp\theta^\bullet,t)$ is either $0$ or $V(\psi^\bullet,t)$, which by the induction hypothesis is an element of $\{0,1\}$. 		
 		\item[\sc Case $\varphi=\odot \psi$,  $\odot\in \{\nx,\ps,\nec\}$:] Consider $\varphi=\nec\psi$, as the other cases are similar.
 		By definition we have that $V((\nec \psi)^\bullet, \allowbreak t) = V(\nec (\psi^\bullet), t) = \inf_{n<\omega}  V(\psi^\bullet, S^n(t))$.
 		Since by induction every $V(\psi^\bullet, t') \in \lbrace 0,1 \rbrace$, their infimum belongs to $\lbrace 0,1 \rbrace$ as well.\qed
 	\end{description}
 \end{proof}

 We introduce the following {\sf GTL}/{\sf LTL} model correspondence: given an {\sf LTL} model $(\mathbb{N},S,V)$,
 we associate a crisp G\"odel model $(\mathbb{N},S,V')$ where $V'(p, t) = 1$ if $t \in V(p)$ and $0$ otherwise. We can prove by induction the following result.
 
 \begin{prop}\label{prop1} For any $\varphi \in \mathcal L$ and for all $t \in \mathbb{N}$, 
 	\begin{enumerate}
 		\item if $(\mathbb{N},S,V), t \models \varphi$ then $V'(\varphi^\bullet,t) = 1$;
 		\item if $(\mathbb{N},S,V), t \not \models \varphi$ then $V'(\varphi^\bullet,t) = 0$.
 	\end{enumerate}
 \end{prop}
 \begin{proof}By structural induction. Left to the reader.\end{proof}	
 
 Conversely, given a real G\"odel temporal model $(\mathbb{N},S,V')$, we associate the (crisp) model $(\mathbb{N},S,V)$ by fixing $V(p, t) = V'(\neg \neg p, t) \in \lbrace 0,1 \rbrace$. The following result can be easily obtained by structural induction.
 \begin{prop}\label{prop2} For any $\varphi \in \mathcal L$ and any $t \in \mathbb{N}$, we have $V'(\varphi^\bullet,t)=V(\varphi,t)$.\end{prop}
 %	\begin{itemize}
 	%		\item if $V'((\varphi)^N,t) =1$ then $(\mathbb{N},S,V), t \models \varphi$
 	%		\item if $V((\varphi)^N,t) \not= 1$ then $(\mathbb{N},S,V), t \not \models \varphi$  	
 	%	\end{itemize}
 %\end{document}	

 \begin{proof} By structural induction. 
The case $\varphi\in\mathbb P$ follows from $V'(p^\bullet,t)  = V'(\neg \neg p,t)\allowbreak = V(p,t)$ by definition, and other cases follow from $(\cdot)^\bullet$ being homomorphic and the real semantics coinciding with classical truth definitions when values are in $\{0,1\}$. %\david{I think this sketch suffices.}
 \end{proof}

 As a corollary we get the following. 
 
 \begin{corollary} For any $\varphi \in \mathcal L$, we have ${\sf LTL} \models \varphi$ if and only if ${\sf GTL} \models \varphi^\bullet$. \end{corollary}

 \begin{proof}
 	For the left-to-right direction, assume by contraposition that ${\sf {\sf GTL}}\not \models \varphi^\bullet$. Therefore there exists a G\"odel temporal model $(\mathbb{N},S,V')$ and $t \in \mathbb{N}$ such that $V'(\varphi^\bullet,t) \not= 1$. By proposition~\ref{prop1} there exists a crisp G\"odel temporal model $(\mathbb{N},S,V)$ such that $V(\varphi,t) \not= 1$. This latter model can be turned into an ${\sf LTL}$ model. Therefore, ${\sf LTL}\not \models \varphi$.
 	
 	Conversely, assume by contraposition that ${\sf LTL}\not \models \varphi$. This means that there exists an {\sf LTL} model $M$ and $t \in \mathbb{N}$ such that $M, t \not \models \varphi$. Then $M$ can be turned into a \emph{crisp} G\"odel temporal model $(\mathbb{N},S,V)$ such that $V(\varphi,t) = 0$. As a consequence ${\sf GTL}\not \models \varphi$.
 \end{proof}		 
 
 For the {\sc pspace}-inclusion, we adapt the proof of {\sf LTL} satisfiability from~\cite{degola16a} to the case of {\sf GTL}.
 Say that an \define{ultimately periodic quasimodel} is a quasimodel $\mathcal Q=(W ,{\leq},\ell,S)$ such that
 there is a flow $(T,f)$ with $T=\{0,\ldots,i+l\}$ with $f(k)=k+1$ for $k<i+l$, $f(i+l) = i$, and a projection function $\pi\colon W\to T$ such that $\pi^{-1}(t)$ is a linear component of $W$ and $w \mathrel S v$ implies that $\pi(v)=f(\pi(w))$.
% We say that the \define{initial lenght} of $\mathcal Q$ is $i$ and its \define{loop lengh} is $l$.

In other words, $\mathcal Q$ has an underlying flow $T$ consisting of an initial segment followed by a loop, and each $t\in T$ is assigned a linear order $\pi^{-1}(t)$, which we may also write as $W_t$.
Every falsifiable formula is falsifiable in a quasimodel of this form.
  
 \begin{thm}[ultimately periodic quasimodel property]\label{aperiodic} Every falsifiable $\mathcal L$-formula is falsifiable in an ultimately periodic quasimodel of height bounded by $|\Sigma|+1$.
 \end{thm}
 
 \begin{proof}
 We sketch the construction.
 By Lemma \ref{quasi:bound}, if $\varphi$ is falsifiable, it is falsifiable on some quasimodel $\mathcal Q' = (W',{\leq}'  ,\ell' ,S' )$ of height at most $\lgt \Sigma+1$.
 Choose $w_0 \in W' $ such that $\varphi\notin\ell_0(w_0)$, and let $W_0$ be the linear component of $w_0$ (i.e., $W_0 = \{v\in W \mid v\leq w_0\text{ or }w_0\leq v\}$) and $\leq_0$ be $\leq'$ restricted to $W_0$.
 By a priority method similar to that of Section \ref{SecGen}, we define an infinite sequence $(W_0,{\leq}_0),(W_1,{\leq}_1),\ldots,$ and sensible relations $S_k\subseteq W_k\times W_{k+1}$, such that $\mathcal Q^{\infty} = (W^\infty,
 {\leq}^\infty,\ell^\infty,S^\infty) $ is a quasimodel, where $ W^\infty =\bigsqcup_{k<\omega} W_k$ ($\bigsqcup$ denotes a disjoint union),  $ {\leq ^\infty} = {\bigsqcup_{k<\omega} \leq_k}$, and so on. %\david{This almost certainly needs more detail.}

 Note that there are at most $2^{\lgt \Sigma(\lgt\Sigma+1)}$ possible choices of $W_k$, since each $W_k$ consists of at most $\lgt\Sigma+1 $ types, and there are at most $2^{\lgt\Sigma}$ types.
 This in particular implies that some $W_i$ repeats infinitely often.
 Let $l$ be such that $W_{i+l} = W_i$ and every defect of $W_i$ has been realised before $W_{i+l}$; such an $l$ exists because $W_i$ has finitely many $\Diamond$ or $\Box$ defects.
 We define $\mathcal Q=(W,{\leq},\ell,S)$ to be the restriction of $\mathcal Q^\infty$ to $\bigcup_{k=0}^{i+l-1}W_k$, but with $S$ redefined on $W_{i+l-1}$ so that it maps to $W_i$. 

It remains to check that $\mathcal Q$ is a quasimodel.
We only check that it is $\omega$-sensible, as the other properties are easy to check.
Consider the case of $\Diamond \psi\in \ell(w)$ (the case $\Box\psi\in\Sigma\setminus \ell(w)$ is analogous).
Then $w\in W_k$ for some $k$, which means that for some $j$ (namely, $j=i+l-k$), there is $v\in W_i$ such that $w\mathrel S^j v$.
By construction, there are some $k'$ and some $u$ such that $v\mathrel S^{k'} u$ and $\psi\in \ell(u)$.
Then $w\mathrel S^{k+k'}u$ and $\psi\in\ell(u)$, as needed.
 \end{proof}

Ultimately periodic models can be represented using {\em moments.}

\begin{definition}
A \define{$\Sigma$-moment} is a sequence of the form $\mathfrak m = (\mathfrak m_0,\ldots,\mathfrak m_{m} )$, %\david{I made the lengths be variable, since it is more consistent with the rest of the paper.}
 where 
\begin{enumerate}
\item
each $\mathfrak m_i$ is a $\Sigma$-type, 
\item\label{two}
$\mathfrak m_i\supsetneq \mathfrak m_{i+1}$ for $i<n-1$, 
\item\label{three}
 for every $\varphi \imp \psi \in \Sigma\setminus \mathfrak m_{i} $ there is some $j\leq i$ with $\varphi \in \mathfrak m_{j}$ but $\psi \not\in \mathfrak m_{j}$,
 \item\label{four}
 for every $\varphi \dimp \psi \in  \mathfrak m_{i} $ there is some $j\geq i$ with $\varphi \in \mathfrak m_{j}$ but $\psi \not\in \mathfrak m_{j}$.
 \end{enumerate}
 We write $|\mathfrak M|$ for the set $\{\mathfrak m_0,\ldots,\mathfrak m_m\}$.
The set of $\Sigma$-moments is denoted $M_\Sigma$.
\end{definition}

%\begin{lem}\label{lem:space}
%The triple $\<M_\Sigma,{\leq_\Sigma},\ell_\Sigma\>$ is a $\Sigma$-labelled space.
%\end{lem}

%\david{Adapt this.}
%It is clear from the definitions that $\<M_\Sigma,{\leq_\Sigma}\>$ is locally linear and that $\ell_\Sigma$ is inversely monotone. 
%\color{red}
%It remains to show that for all $\mathfrak m \in M_\Sigma$, whenever $\varphi \imp \psi \in \defect{ \ell_  \Sigma(\mathfrak m)}$, there is $\mathfrak n \leq_\Sigma \mathfrak m$ such that $\mathfrak n$ revokes $\varphi \imp \psi$. %First we check that $\ell_\Sigma$ is continuous. This amounts to showing that, if $\fw\geq_\Sigma\fv\in M_\Sigma$, then $\ell_\Sigma(\fw)\subseteq \ell_\Sigma(\fv)$. But, $\fw\geq_\Sigma\fv$ means that $\fv=\fw[v]$ for some $v\leq_\fw \Root\fw$, hence $\ell_\Sigma(\fv)=\ell_\fw (v)\supseteq \ell_\fw(\Root \fw)=\ell_\Sigma(\fw)$ by the continuity of $\ell_\fw$. Similarly, 
%But if  $\varphi \imp \psi\in \defect{ \ell_  \Sigma(\mathfrak m)}$, then by condition \eqref{three} there is an $\mathfrak m_i$ with $\varphi \in \mathfrak m_{i-1}$ but $\psi \not\in \mathfrak m_{i-1}$; but then $(\mathfrak m_0, \dots, \mathfrak m_i) \leq_\Sigma \mathfrak m$ and by the definition of $\ell_\Sigma$ the moment $(\mathfrak m_0, \dots, \mathfrak m_i)$ revokes $\varphi \imp \psi$.
%\color{black}
%\endproof

% We now wish to define a $\Sigma$-labelled \emph{system} $\moments\Sigma$ over $M_\Sigma$. For this, we must define a fully confluent convex sensible relation on $\<M_\Sigma,{\leq_\Sigma},\ell_\Sigma\>$.
 
We define the labelled space $(\mathfrak m_0,\ldots,\mathfrak m_{m} ) + (\mathfrak n_0,\ldots,\mathfrak n_{n} )$ to be the parallel sum of the two linear posets $(\mathfrak m_0,\ldots,\mathfrak m_{m} )$ and $(\mathfrak n_0,\ldots,\mathfrak n_{n})$ with labelling given by the identity.

\begin{definition}\label{ts}
Say the moment $\mathfrak n$ is a \define{temporal successor} of $\mathfrak m$, denoted $\mathfrak m \mathrel S_\Sigma \mathfrak n$, if there exists a fully confluent convex sensible relation
$R\subseteq |\mathfrak m| \times |\mathfrak n|$ on the labelled space $\mathfrak m + \mathfrak n$.
\end{definition}

%\begin{definition}
% We set
%$\moments\Sigma=\<M_\Sigma,{\leq_\Sigma},\ell_\Sigma,S_\Sigma\>$.
%\end{definition}

\begin{definition}
 We define
$\moments\Sigma=\<M_\Sigma, S_\Sigma\>$.
\end{definition}

%\begin{lem}\label{LemmIsSensible}
%$\moments\Sigma$ is a $\Sigma$-labelled system.
%\end{lem}

%That $\<M_\Sigma,{\leq_\Sigma},\ell_\Sigma\>$ is a $\Sigma$-labelled space is \ref{lem:space}, so it only remains to show that $S_\Sigma$ is a fully confluent convex sensible relation.

%Let $\mathfrak m \mathrel S_\Sigma \mathfrak n$. Then $\mathfrak m_{|\mathfrak m|-1} \mathrel R \mathfrak n_{|\mathfrak n|-1}$ for some sensible $R\subseteq \{\mathfrak m_0,\ldots,\mathfrak m_{|\mathfrak m|-1} \}\times \{\mathfrak n_0,\ldots,\mathfrak n_{|\mathfrak n|-1} \}$, and hence the pair $(\mathfrak m_{|\mathfrak m|-1}, \mathfrak n_{|\mathfrak n|-1})$ is sensible. By the definition of $\ell_\Sigma$, the pair $(\ell_\Sigma(\mathfrak m), \ell_\Sigma(\mathfrak n))$ is then sensible. Hence $S_\Sigma$ is sensible.

%Next we verify that $S_\Sigma$ is fully confluent. The forward-down and backward-down properties follow straightforwardly using the existence of $R$ given by \Cref{ts}. [upwards confluence]

%[convex]\david{To do.}
%\endproof

Because of condition \eqref{two}, if $\Sigma$ is finite then so is $\moments\Sigma$.

 \begin{definition} A \emph{small falsifiability %\david{Changed.}
  witness} for an $\mathcal L$-formu\-la $\varphi$ is a finite sequence of moments $\mathfrak{m}^0,\dots, \mathfrak{m}^i, \dots, \mathfrak{m}^{i+l}$ of subsets of $\Sigma$ with a distinguished position $i$ and binary relations $S_j \subseteq |\mathfrak m^j| \times |\mathfrak m^{j+1}|$ for each $j<i+l$ such that 
 	\begin{enumerate}[label=(\Alph*)]
 	%	\item $0 \le i \le 2^{\lvert \varphi \rvert}$\david{Are these bounds correct?} and $1 \le l \le \lvert\varphi\rvert 2^{\lvert \varphi \rvert}$.

 		\item $\varphi \in \mathfrak{m}_0^0$ and $\mathfrak{m}^i = \mathfrak{m}^{i+l}$,

 		\item $S_j$ is a sensible relation on $\mathfrak{m}^j + \mathfrak{m}^{j+1}$,

 		\item If $\ps \psi \in \mathfrak m^i_j$ then there are $l'<l$ and a sequence $(j_k)_{k\leq l'}$ with $j_0=j$ such that $\mathfrak m^{i+k}_{j_k} \mathrel S_{i+k} \mathfrak m^{i+k+1}_{j_{k+1}} $ if $k<l'$ and $\psi\in \mathfrak m^{i+l'}_{j_{1'}}$,
 		
 		\item If $\nec \psi \in\Sigma\setminus \mathfrak m^i_j$ then there are $l'<l$ and a sequence $(j_k)_{k\leq l'}$ with $j_0=j$ such that $\mathfrak m^{i+k}_{j_k} \mathrel S_{i+k} \mathfrak m^{i+k+1}_{j_{k+1}} $ if $k<l'$ and $\psi\notin \mathfrak m^{i+l'}_{j_{1'}}$.
 				
 	\end{enumerate}
 \end{definition}

As we will see below, Theorem~\ref{aperiodic} implies that if an $\mathcal L$-formula is falsifiable then it has a small falsifiability witness. Moreover, the converse is also true. As a consequence, we obtain an equivalence between the existence of an infinite structure (a model of $\varphi$) and the existence of a finite structure (a small falsifiability witness) for a given $\mathcal L$-formula $\varphi$.
 
 \begin{thm}	\label{witness}
 	An $\mathcal L$-formula is falsifiable if and only if it has a small falsifiability witness.	
 \end{thm}	
 \begin{proof} 	
 		For the left-to-right direction, assume that the formula $\varphi$ is falsifiable.
 		By Theorem~\ref{aperiodic}, there exists an ultimately periodic quasimodel $\mathcal{Q} = (W,{\leq},\ell,S)$ such that $W=\bigcup_{k<i+l} W_k$ and $\varphi \in \ell(w)$ for some $w\in W_0$.
 		For each $W_k$, let $W_k = \{v^k_0\ldots v^k_{m_k}\}$ in increasing order, and let $\mathfrak m^k = ( \ell (v^k_0) \ldots \ell (v^k_{m_k}) )$ (deleting repeating types if needed).   	
 		It is easy to check that the sequence $\mathfrak{m}^0,\dots,\mathfrak{m}^i, \dots, \mathfrak{m}^{i+l}$ yields a small falsifiability witness.
 		 	
 	Conversely, we will show that if a formula has a small falsifiability witness $\mathfrak{m}^0, \dots, \mathfrak{m}^i, \dots, \mathfrak{m}^{i+l}$ then it is falsifiable.
 	Write $\mathfrak{m}^{k} = (\mathfrak{m}^{k}_0,\ldots,\mathfrak{m}^{k}_{m_k})$ and consider the labelled space $\mathcal{Q} = (W,{\leq}, \ell, S)$, where $W= \{ (\mathfrak m_s^k,k) \mid k<i+l\text{ and } s\leq m_k \}$ and $\leq,\ell,S$ are defined in the obvious way. 
It is not hard to check that $\mathcal{Q} $ is a quasimodel falsifying $\varphi$. Hence by Theorem \ref{sound}, $\varphi$ is falsifiable.
 \end{proof}	
 
 \begin{algorithm}[h!]
 	\caption{{\sf {\sf GTL}} falsifiability algorithm}\label{gtlsat}
 	\begin{algorithmic}
 	\State input $\varphi$
 		\State set $\Sigma$ to be the set of subformulas of $\varphi$
 		\State guess two moments, $\mathfrak{m}$ and $\mathfrak m_f$, of heights $m,n\leq |\Sigma|+1$ such that $\varphi\notin \mathfrak m_m$
% 		\State $j \leftarrow 0$ 		
 		\While{$\mathfrak m\neq \mathfrak m_f$}
 		\State guess a moment $\mathfrak{n}$ of height at most $|\Sigma|+1$
 		\If{$\mathfrak n$ is a temporal successor of $\mathfrak m$}
 		\State $\mathfrak{m} \leftarrow \mathfrak{n}$
 		\Else
 		\State reject
 		\EndIf
% 		\State $j \leftarrow j+1$ 		
 		\EndWhile
 		\State $\Delta_\ps \leftarrow \lbrace (k,\ps \psi) \mid \ps \psi \in  \mathfrak{m}_k,\; 0 \le k <n \rbrace$
 		\State $\Delta'_{\ps} \leftarrow \lbrace (k,\ps \psi) \mid \psi \in  \mathfrak{m}_k,\; 0 \le k < n \rbrace$
 		\State $\Delta_{\nec} \leftarrow \lbrace (k,\nec \psi) \mid \nec \psi \in \Sigma\setminus {\mathfrak{m}}_k,\; 0 \le k < n \rbrace$
 		\State $\Delta_{\nec}' \leftarrow \lbrace (k,\nec \psi) \mid \psi \in \Sigma\setminus{\mathfrak{m}}_k,\; 0 \le k < n \rbrace $ 			
 		\State $S^* \leftarrow \lbrace (k,k)  \mid 0 \le k < n \rbrace.$ 		 		 		
% 		\State $j \leftarrow 0$
 		\While{$\mathfrak{m}_f \neq \mathfrak{m}$ or $\Delta_{\ps} \not \subseteq \Delta'_{\ps}$ or $\Delta_{\nec} \not \subseteq \Delta'_{\nec}$} 

 		\State guess a moment $\mathfrak{n}$ of height at most $|\Sigma|+1$.
 		\If {$\mathfrak n$ is a temporal successor of $\mathfrak m$}
 		\State $\mathfrak{m} \leftarrow \mathfrak{n}$
 		\Else
 		\State reject
 		\EndIf

 		\State $S^* = \lbrace (k,z) \mid (k,y)  \in S^* \hbox{ and } (\mathfrak{m}_y,\mathfrak{n}_z) \in R\rbrace$ 			
 		\State $\Delta'_{\ps} \leftarrow \Delta'_{\ps} \cup \lbrace  (k,\ps \psi) \mid (k,x) \in S^* \hbox{ and } \psi \in  \mathfrak{n}_x\rbrace$
 		\State $\Delta'_{\nec} \leftarrow \Delta'_{\nec} \cup \lbrace  (k,\nec \psi) \mid (k,x) \in S^* \hbox{ and } \psi \in \Sigma \setminus {\mathfrak{n}_x}\rbrace$
 		\State $j \leftarrow j + 1$
% 		\State $\mathfrak{m} \leftarrow \mathfrak{n}$	
 		\EndWhile
% 		\If{ $\mathfrak{m}^f = \mathfrak{m}$ and $\Delta_{\ps} \subseteq \Delta'_{\ps}$ and $\Delta_{\nec} \subseteq \Delta'_{\nec}$}
% 		\State accept
% 		\Else 
% 		\State abort
% 		\EndIf
\State accept
 	\end{algorithmic}
 \end{algorithm}

 \begin{thm}
 	Algorithm~\ref{gtlsat} for falsifiability checking of an $\mathcal L$-formula is correct and works in space that is polynomial in the size of the input formula.
 \end{thm}
 \begin{proof} Completeness follows from Theorem~\ref{witness}.
 If $\varphi$ is falsifiable, then $\varphi$ has a small falsifiability witness $\mathfrak m_0,\ldots,\mathfrak m_{i+l}$ with $\mathfrak m_{i+l} = \mathfrak m_i$.
 We initialise $\mathfrak m$ to $\mathfrak m_0$, $\mathfrak m_f$ to $\mathfrak m_i$, and at step $j$ choose $\mathfrak n$ to be $\mathfrak m_{j+1}$.
 This yields an accepting computation of Algorithm~\ref{gtlsat}.
 
 Conversely, if Algorithm~\ref{gtlsat} has an accepting computation, let $\mathfrak m_0,\ldots,\mathfrak m_{i+l}$ enumerate the values taken by $\mathfrak m$, where $i$ is the least index such that $\mathfrak m_i=\mathfrak m_f$.
 It is not hard to check that this sequence yields a small falsifiability witness.

 In order to check that the nondeterministic algorithm uses polynomial space, it is sufficient to observe that each subset of $\Sigma$ can be encoded by a polynomial number of bits.
 Since $\mathfrak{m}$, $\mathfrak{m}_f$ and $\mathfrak{n}$ have at most length $\lvert \Sigma \rvert$ we need $3 \lvert \Sigma \rvert + 4$ of those sets ($\mathfrak{m}$, $\mathfrak{m}_f$, $\mathfrak{n}$, $\Delta_\ps$, $\Delta'_\ps$, $\Delta_\nec$, and $\Delta'_\nec$).
Checking that $\mathfrak n$ is a temporal successor of $\mathfrak m$ can be done nondeterministically by guessing a relation $R$ and checking that it is a sensible relation; but the size of $R$ is bounded by the product of the sizes of $\mathfrak m$ and $\mathfrak n$.
Similarly,  $\lvert S^* \rvert$ has at most $|\Sigma|^2$ elements so also polynomial.
 	 \end{proof}
 	 
 	  	According to Savitch's theorem~\cite{SAVITCH1970177}, nondeterministic polynomial space is included in deterministic polynomial space.
 	Applying this to falsifiability checking yields the following complexity upper bound.

 \begin{corollary}
 	The decision problem of testing falsifiability for ${\sf GTL}$ is {\sc {\sc pspace}}-complete.
 \end{corollary}

 \section{Concluding remarks}
 
 We have defined a natural version of linear temporal logic based on a G\"odel--Dummett base and shown that it may equivalently be characterized as a fuzzy logic or as a superintuitionistic logic using standard semantics in each case.
 Despite the lack of a finite model property for either of the two semantics, we have introduced a class of quasimodels for which $\sf GTL$ {\em does} satisfy a version of the finite model property, and moreover shown how they can be used to adapt the classical proof of {\sc {\sc pspace}}-completeness for $\sf {\sf LTL}$.

This puts G\"odel temporal logics in sharp contrast to other fuzzy logics, whose transitive modal logics are undecidable~\cite{Vidal21}, or intuitionistic temporal logics, where systems are known to be decidable only with non-elementary upper bounds, if at all \cite{F-D18,BalbianiToCL}.
This places G\"odel--Dummett logic as {\em the} premier base for computational applications of sub-classical modal and temporal logics.

The techniques we have used are quite robust and should readily generalise to logics such as $\sf PDL$ or even the G\"odel $\mu$-calculus.
This represents a milestone in the program pioneered by Caicedo et al.~\cite{CaicedoMRR17} of extending complexity results from classical modal and temporal logics to their G\"odel counterparts.

%\subsubsection{Acknowledgements} Please place your acknowledgments at
%the end of the paper, preceded by an unnumbered run-in heading (i.e.
%3rd-level heading).

%
% ---- Bibliography ----
%
% BibTeX users should specify bibliography style 'splncs04'.
% References will then be sorted and formatted in the correct style.
%
% \bibliographystyle{unsrt}
% \bibliography{mybibliography}
%
%\bibliography{biblio}

\end{document}